\pdfoutput=1					

\documentclass[%
	12pt,
	numbers=noenddot,
	bibliography=totoc,
]{scrartcl}

\usepackage[UKenglish]{babel}
\usepackage{amsfonts}
\usepackage{amsmath}
\usepackage{amssymb}
\usepackage{array}
\usepackage{booktabs}
\usepackage{threeparttable}
\usepackage{caption}
\usepackage{comment}
\usepackage{dcolumn}
\usepackage{enumitem}
\usepackage{epigraph}
\usepackage{epsfig}
\usepackage{epstopdf}
\usepackage{float}
\usepackage{footmisc}
\usepackage{geometry}
\usepackage{graphicx}
\usepackage{graphics}
\usepackage{listings}
\usepackage{makecell}
\usepackage{mathrsfs}
\usepackage{mathtools}
\usepackage{marvosym}
\usepackage{multirow}
\usepackage{nicefrac}
\usepackage{pdflscape}
\usepackage{placeins}
\usepackage{rotating,tabularx}
\usepackage{setspace}
\usepackage{subcaption}							
\usepackage[normalem]{ulem}
\usepackage[dvipsnames]{xcolor}
\usepackage{xspace}

\usepackage[color=blue!30]{todonotes}

\usepackage[
 	breaklinks=true,
 	urlcolor=blue,															
  colorlinks=true,
  linkcolor=blue,
  citecolor=blue,
]{hyperref}

\usepackage{url,doi} 											

\setkomafont{captionlabel}{\sffamily\bfseries}
\setcapindent{1em}

\usepackage[
arxiv=abs,																					
autocite=inline,																		
autolang=hyphen,																		
backend=bibtex8,																		
backref=true,																				
bibstyle=authoryear,																
citestyle=authoryear-comp,													
dashed=false,																				
sortcites=false,																		
maxbibnames=11,																			
maxcitenames=2,																			
mincitenames=1,																			
url=true,doi=true,eprint=true,											
useprefix=true,																			
]{biblatex}

\newcommand*\person[1]{\textsc{#1}}									

\DeclareFieldFormat{pages}{#1}		

\renewbibmacro{in:}{}					

\usepackage{tikz}

\hypersetup{
pdfauthor={Yonatan Berman, Mark Kirstein},
pdftitle={Risk Preferences in Time Lotteries},
pdfkeywords={Ergodicity Economics, Growth-Optimality, Time Lotteries},
}

\usepackage{amsthm}
\usepackage{etoolbox}

\newtheoremstyle{myDefStyle}
	{\topsep} 						
	{\topsep} 						
	{\itshape} 						
	{} 										
	{\sffamily\bfseries}	
	{} 										
	{\newline} 						
	{\thmname{#1} \thmnumber{#2} \thmnote{{\bfseries(#3)}}}	

\theoremstyle{myDefStyle}
\newtheorem{definition}{Definition}

\newtheorem{proposition}{Proposition}

\newcommand{\ra}[1]{\renewcommand{\arraystretch}{#1}}

\newcolumntype{d}[1]{D{.}{.}{#1}} 

\newcommand{\red}[1]{{\color{red} #1}}

\renewcommand*{\qed}{\hfill\ensuremath{\blacksquare}}%

\newcommand{\ie}{\textit{i.e.}\xspace}

\newcommand{\EA}[1]{\ensuremath{\left\langle #1 \right\rangle}\xspace}		
\newcommand{\EV}[1]{\ensuremath{\operatorname{E}\left[ #1 \right]}\xspace}	

\newcommand{\gr}{\ensuremath{g}\xspace}

\newcommand{\EAgr}{\ensuremath{\EA{\gr}}\xspace}
\newcommand{\EAgrA}{\ensuremath{\EA{\gr_{\mathcal{A}}}}\xspace}
\newcommand{\EAgrB}{\ensuremath{\EA{\gr_{\mathcal{B}}}}\xspace}
\newcommand{\EAgra}{\ensuremath{\EA{\gr_a}}\xspace}
\newcommand{\EAgrb}{\ensuremath{\EA{\gr_b}}\xspace}
\newcommand{\EAgrc}{\ensuremath{\EA{\gr_c}}\xspace}

\newcommand{\TAgr}{\ensuremath{\bar{\gr}}\xspace}

\newcommand{\TAgrA}{\ensuremath{\bar{\gr}_{\mathcal{A}}}\xspace}
\newcommand{\TAgrB}{\ensuremath{\bar{\gr}_{\mathcal{B}}}\xspace}

\newcommand{\TAgra}{\ensuremath{\bar{\gr}_a}\xspace}
\newcommand{\TAgrb}{\ensuremath{\bar{\gr}_b}\xspace}
\newcommand{\TAgrc}{\ensuremath{\bar{\gr}_c}\xspace}

\newcommand{\ETa}{\ensuremath{\langle t_a\rangle}\xspace}
\newcommand{\ETb}{\ensuremath{\langle t_b\rangle}\xspace}
\newcommand{\ETc}{\ensuremath{\langle t_c\rangle}\xspace}

\newcommand{\nn}{\nonumber}
\newcommand{\natop}{\ensuremath{\alpha}\xspace}								

\newcommand{\TP}{\ensuremath{\mathcal{M}}\xspace}			
\newcommand{\TL}{\ensuremath{L}\xspace}														
\newcommand{\DTL}{\ensuremath{\mathcal{L}}\xspace}									

\newcommand{\D}{\ensuremath{\Delta}\xspace}

\newcommand{\TLa}{\ensuremath{\TL_a}\xspace}							
\newcommand{\TLb}{\ensuremath{\TL_b}\xspace}							
\newcommand{\TLc}{\ensuremath{\TL_c}\xspace}							

\newcommand{\gra}{\ensuremath{\gr_a}\xspace}			
\newcommand{\grb}{\ensuremath{\gr_b}\xspace}			
\newcommand{\grc}{\ensuremath{\gr_c}\xspace}			

\newcommand{\Cref}[1]{Corollary~\ref{coro:#1}}
\newcommand{\cref}[1]{Cor.~\ref{coro:#1}}

\newcommand{\dlabel}[1]{\label{def:#1}}

\newcommand{\elabel}[1]{\label{eq:#1}}
\newcommand{\eref}[1]{Eq.~(\ref{eq:#1})}

\newcommand{\flabel}[1]{\label{fig:#1}}
\newcommand{\fref}[1]{Fig.~\ref{fig:#1}}

\newcommand{\propref}[1]{Proposition~\ref{prop:#1}}

\newcommand{\tlabel}[1]{\label{tab:#1}}
\newcommand{\tref}[1]{Tab.~\ref{tab:#1}}
\newcommand{\Tref}[1]{Table~\ref{tab:#1}}

\newcommand{\seclabel}[1]{\label{sec:#1}}

\newcommand{\secref}[1]{Section~\ref{sec:#1}}
\newcommand{\sseclabel}[1]{\label{ssec:#1}}
\newcommand{\ssref}[1]{Subsec.~\ref{ssec:#1}}

\newcommand{\be}{\begin{equation}}
\newcommand{\ee}{\end{equation}}
\newcommand{\bee}{\begin{equation}}
\newcommand{\eee}{\end{equation}}
\newcommand{\ba}{\begin{align}}
\newcommand{\ea}{\end{align}}
\newcommand{\bc}{\begin{center}}
\newcommand{\ec}{\end{center}}
\newcommand{\bea}{\begin{eqnarray}}
\newcommand{\eea}{\end{eqnarray}}
\newcommand{\bi}{\begin{itemize}}
\newcommand{\ei}{\end{itemize}}

\newcommand{\tI}{\ensuremath{t_1}\xspace}

\newcommand{\tII}{\ensuremath{t_2}\xspace}
\newcommand{\Dt}{\ensuremath{\Delta t}\xspace}

\newcommand{\gI}{\ensuremath{g_{\TP_1}}\xspace}
\newcommand{\gII}{\ensuremath{g_{\TP_2}}\xspace}

\newcommand{\DTLgr}{\ensuremath{g_{\DTL}}\xspace}

\newcommand{\oI}{\text{I}\xspace}
\newcommand{\oII}{\text{II}\xspace}

\newcommand{\Dx}{\ensuremath{\Delta x}\xspace}

\numberwithin{equation}{section}

\title{Risk Preferences in Time Lotteries\thanks{~~We benefited from detailed comments by \person{Alexander Adamou} and \person{Ole Peters}. We also wish to thank \person{Alex Kacelnik} for fruitful discussions, along with seminar participants at the Max Planck Institute for Mathematics in the Sciences, Leipzig, University of Duisburg-Essen, and Rutgers University, and with conference participants at the D-TEA Conference 2020 and the 2020 Annual Meeting of the German Economic Association.
}}

\author{
	\person{Yonatan Berman}\footnote{\texttt{\href{mailto:y.berman@lml.org.uk}{~\Letter~y.berman@lml.org.uk}}, London Mathematical Laboratory} \and 
	\person{Mark Kirstein}\footnote{\texttt{\href{mailto:m.kirstein@mis.mpg.de}{~\Letter~m.kirstein@mis.mpg.de}},
	Max-Planck-Institute for Mathematics in the Sciences
	}
}

\begin{document}

\maketitle

\begin{abstract}
\noindent An important but understudied question in economics is how people choose when facing uncertainty in the timing of events.
Here we study preferences over time lotteries, in which the payment amount is certain but the payment time is uncertain.
Expected discounted utility theory (EDUT) predicts decision makers to be risk-seeking over time lotteries.
We explore a normative model of growth-optimality, in which decision makers maximise the long-term growth rate of their wealth.
Revisiting experimental evidence on time lotteries, we find that growth-optimality accords better with the evidence than EDUT.
We outline future experiments to scrutinise further the plausibility of growth-optimality.
\end{abstract}

\vspace{1em}

\vspace{.5em}

\noindent\textsf{\textbf{JEL Classification\ }}
\href{https://www.aeaweb.org/econlit/jelCodes.php?view=jel#C}{
C61}	
$\cdot$
\href{https://www.aeaweb.org/econlit/jelCodes.php?view=jel#D}{%
D01 	
$\cdot$
D81 	
$\cdot$
D9} 	







\section{Introduction}
%

Real-world economic decisions are usually a combination of at least two types of uncertainties.
First, uncertainty about the exact payment \textit{amount}, and second, uncertainty about the exact payment \textit{time}.
Examples for the latter include
investing in an $\text{R}\&\text{D}$ venture with unknown completion date,
exercising an American option,
deciding on an optimal payment schedule of a life assurance with a specified death benefit,
or choosing a delivery at a random or a guaranteed date.
Uncertainty about the exact timing of events also abounds in many decisions that involve expenditures or losses like climate catastrophes with hard-to-predict timing but more or less known costs.
While uncertainty about payment time is crucial for many real-world problems, both theoretical and empirical literatures on decision-making are mostly focused on random payment amounts.

A recent paper by \textcite{DeJarnetteETAL2020} coined the term \textit{time lotteries}. These are lotteries in which the payment amount is known with certainty, but the payment time is random. Expected Discounted Utility theory (EDUT) predicts decision makers to be risk-seeking over time lotteries -- they would choose the lottery that corresponds to higher uncertainty in the payment time (for the same expected payment time).\footnote{Like \textcite{DeJarnetteETAL2020}, ``we interpret outcomes [\ldots] as representing goods that are consumed on date $t$ rather th[a]n prizes received in a certain date that need not coincide with the time of actual consumption.'' We thus follow the Money Earlier or Later framework \parencite{CohenETAL2020}, keeping in mind the existing controversy in the literature regarding its relevance to decisions made in the context of consumption rather than money payments.}

Yet, the experimental evidence on decision makers' risk preferences in time lotteries contrasts this prediction. \textcite{ChessonViscusi2003,OnayOnculer2007,Ebert2018,DeJarnetteETAL2020} all find heterogeneous results that contradict the prediction of EDUT. At most, decision makers in experiments generally preferred lotteries with lower uncertainty in payment time.

This paper studies risk preferences in time lotteries by exploring a normative decision model that differs from EDUT in its optimand: the growth rate of the decision maker's wealth. More specifically, we study time lotteries under growth-optimality, \ie when decision makers maximise the long-term growth rate of their wealth. The paper further aims to reconcile the experimental evidence with a theory of choice.



In time lotteries, the growth rate of wealth is a random variable. To make a growth-optimal choice between time lotteries it is thus necessary to collapse the random growth rate into a scalar. This can be done by taking different averages, leading to two approaches for calculating the growth rate associated with a time lottery.
In the first approach, the scalar growth rate associated with a time lottery is the \textit{time-average} growth rate, experienced by the decision maker if the lottery is indefinitely repeated over time. 
We call this the time approach.
Maximisation of this growth rate predicts risk-neutral behaviour in time lotteries. The time approach is the implementation of ``ergodicity economics'' \parencite{Peters2019b} to derive an optimal decision behaviour from the decision maker's experience over time. Studying time lotteries with the maximisation of the time-average growth rate as a choice criterion is the focus of this paper.

In the second approach, the scalar growth rate is derived as the expected value of the random growth rate or synonymously the \textit{ensemble-average} growth rate. A decision maker experiences this growth rate as if the lottery is simultaneously realised infinitely many times, which might be impossible under most circumstances. We call this approach the ensemble approach. Maximising the ensemble-average growth rate is equivalent to maximising the expected change in discounted utility, \ie to EDUT. It thus predicts risk-seeking behaviour in time lotteries.

We then use the experimental results of \textcite{OnayOnculer2007,DeJarnetteETAL2020} to reconcile the predictions of the two approaches. For the experiment conducted by \textcite{DeJarnetteETAL2020} we find that when the difference between the ensemble-average and time-average growth rates was small, decision makers were neither significantly risk-seeking nor risk-averse. Since the time approach predicts decision makers will be risk-neutral in time lotteries, it accords better with the experimental evidence than the ensemble approach.

As observed in the literature, decision makers' choices substantially deviated from the prediction of the ensemble approach. Yet, the difference between the ensemble-average and time-average growth rates may still be informative on these choices. As the ensemble-average growth rate got larger with respect to the time-average growth rate, decision makers were more likely to prefer less risky lotteries in experiments. Thus, the higher the ensemble-average growth rate was relative to the growth rate associated with a sure payment, the less attractive it became. The reanalysis of the experimental data suggests that decision makers may not consider the ensemble-average growth rate, equivalent to EDUT, as a relevant criterion for their choices. This result is found for both experiments done by \textcite{OnayOnculer2007} and \textcite{DeJarnetteETAL2020}.

Our main contribution is the rationalisation of risk-seeking, risk-neutral and risk-averse behaviour in time lotteries in a normative model with a single choice criterion. We emphasise that growth-optimality satisfies standard axioms of choice \parencite{vonNeumannMorgenstern1944}. It assumes neither behavioural bias nor dynamic inconsistency in the decision maker's behaviour. At all times she prefers the option with the highest growth rate.

This paper also provides a framework for future experiments. It is possible to design choice problems between time lotteries for which one lottery is preferred in the ensemble approach and the other in the time approach. Such experiments, based on the model provided in this paper, are able to confirm or falsify the two approaches, and are planned.

In addition, this paper contributes to a growing body of work in ergodicity economics. Our analysis reveals in a novel way a conceptual problem in the implicit assumption of ergodicity in decision theory \parencite{Peters2019b}, which manifests in using expected values as optimands.
\person{Boltzmann} called the ergodic hypothesis and the interchangeable use of time averages and ensemble averages a ``Kunstgriff'' (trick) because only under ergodicity are the two types of averages identical, such that an expected value captures the behaviour over time \parencite[345]{Boltzmann1872}.
This trick fails in the context of time lotteries. It is impossible to resurrect any decision criterion based on expected values, like in EDUT, if the uncertainty appears in the time domain. Our findings thus reinforce the conceptual insight of ergodicity economics to derive an optimal decision behaviour from the decision maker's experience over time.

The paper is organised as follows. \secref{model} sets out the choice problem and the theoretical framework for analysing time lotteries. In \secref{EvaluatingTL} we describe how growth-optimality is used to evaluate time lotteries and determine the risk preferences predicted by the ensemble and time approaches. \secref{gDiff} reconciles growth-optimality with existing experimental evidence on risk preferences in time lotteries and presents distinguishing experimental setups able to falsify the predictions of the time approach. We conclude in \secref{Discussion} with a discussion on the significance of the results.


\subsection{Related Literature}

\paragraph{Economics}
Research on time lotteries appears under different aliases in the economic literature, such as lottery timing, time ambiguity, uncertain delay \parencite{ChessonViscusi2003}, timing risk \parencite{OnayOnculer2007}, delay risk \parencite{Ebert2018}, time risk \parencite{Ebert2020} and money earlier or later tasks \parencite{CohenETAL2020}.
We adopt the term \textit{Time Lottery} coined recently in \textcite{DeJarnetteETAL2020}.


As mentioned above, EDUT predicts risk-seeking behaviour in time lotteries (RSTL). The experimental studies thus emphasise any finding of deviant risk-averse behaviour in time lotteries (RATL). Indeed, in all of these experiments \parencite{ChessonViscusi2003,OnayOnculer2007,Ebert2018,DeJarnetteETAL2020} a large fraction of subjects were risk-averse. Yet, there were many risk-seeking subjects across all studies. We show how growth-optimality offers an informative analysis of the experimental evidence in \secref{gDiff}.

In an unincentivised survey sent to 373 business owners and managers, \textcite{ChessonViscusi2003} explain the deviations from the EDUT prediction by ambiguity aversion.
\textcite{OnayOnculer2007} performed incentivised, yet inconsequential, experiments with 55 economics undergraduates.
They find a large fraction of RATL decision makers, though not in all the conducted experiments, and especially not when timing risk is large.
They suggest such preferences might be explained by probability weighting.
\textcite{Ebert2018} analysed the choices of 80 students and found substantial heterogeneity of risk preferences over time lotteries.
\textcite{DeJarnetteETAL2020} performed incentivised experiments with 196 students and recruited 156 participants on Amazon Mechanical Turk.
In line with the previous literature, they report consistent violations of the EDUT prediction, and find a significant fraction of RATL decision makers.
To rationalise their findings they propose a generalisation of EDUT, adding additional curvature by wrapping the discount function in another function to absorb the excess risk aversion.

In summary, the remedies offered so far in the literature to the surprising experimental results are incompatible with EDUT. They either suggest behavioural explanations or add free parameters. Overall, however, both theoretical and experimental research on human behaviour in time lotteries is scant, especially when compared to related work on discounting or on uncertainty in payment amount.


\paragraph{Behavioural Ecology}
Contrary to this, in the context of optimal foraging theory in behavioural ecology, animal behaviour in time lotteries has been extensively studied over many species. Models in optimal foraging theory involve different optimands, referred to as currencies. These optimands are usually determined experimentally and differ across species and environmental conditions. In the classical model of optimal foraging the optimand is the long-term average rate of energy gain \parencite[7]{StephensKrebs1986}. Adding constraints to the model renders competing models and optimands plausible, such as maximisation of the optimal use of time, total food uptake, opportunity cost, mean time to reach satiation or specific nutrient maximisation \parencite{Gross1986}. Ultimately, the optimand strictly determines a suitable fitness criterion.


In \secref{EvaluatingTL} we study a model of growth rate maximisation. We present two approaches to find a suitable scalar growth rate. One is equivalent to EDUT, and the other is based on the maximisation of the time-average growth rate. It turns out that both approaches are used to explain risk preferences in behavioural ecology and have well-understood currency analogues:
\bi
\item The time approach of growth rate maximisation (see \ssref{TimePerspective}) corresponds to an optimand known as the ratio of expectations (RoE), \ie the ratio of the expected gain over the expected delay \parencite[\pno~314, eq. 1]{BatesonKacelnik1995}.
\item The standard model in economics, EDUT, follows the ensemble approach, see \ssref{EnsemblePerspective}. This growth rate corresponds to an optimand known as expectation of ratios (EoR) of gains over delay \parencite[\pno~314, eq. 2]{BatesonKacelnik1995} or per patch rate maximisers \parencite[Box 2.1]{StephensKrebs1986}.
\ei
The difference between the two approaches matters greatly for at least two reasons. First, they lead to different values and thus to different predicted behaviours in time lotteries. Second, they highlight the inability of arbitrary growth rates to correspond to real-world experiences. This correspondence is clearly desirable if a model aims to capture not only choices or behaviours, but also mechanisms and logic.

The importance of real-world relevance of any proposed mathematical behavioural model is well understood in behavioural ecology.
For example, \textcite[p.~16, see esp. Box 2.1]{StephensKrebs1986} mention that there is no clear mechanism generating the associated time interval of the ensemble-average growth rate.
This is in line with \ssref{EnsemblePerspective}, where we show that it is possible to construct a growth rate which is equivalent to EDUT preferences, however there is no realistic mechanism that could generate this particular growth rate. 
Furthermore, it is known that the difference between the EoR and the RoE is a consequence of \person{Jensen's} inequality \parencite{Jensen1906,BatesonKacelnik1996}.
\person{Jensen's} inequality reappears in our proof of \propref{rstl} which establishes the equivalence of EDUT and the ensemble approach.
Paradoxical statements result from the erroneous treatment of \person{Jensen's} \textit{inequality} as an \textit{equality} and have led to the \textit{Fallacy of Averages} \parencite{TempletonLawlor1981,Smallwood1996}.

Studies on non-human animals report an extremely stable pattern of exclusively RSTL preferences for a wide range of species \parencite[see the surveys of][]{BatesonKacelnik1995,KacelnikBateson1996}. The contrasting behaviour of humans (varying shares of RATL) and non-human animals (all RSTL) is interesting in its own right, yet it exceeds the scope of this paper.

\paragraph{Ergodicity Economics}
Our work also contributes to the growing field of ergodicity economics, exploring decision-making under the postulate that decision makers maximise the long-time growth rate of resources \parencite{PetersGell-Mann2016,PetersAdamou2018c,BermanPetersAdamou2019,Peters2019b}. The conceptual focus of ergodicity economics is the embedding of uncertainty within time and not in the ensemble.
Despite conceptual differences, ergodicity economics is capable of mechanistically resurrecting EDUT as a special case if the appropriate utility function happens to generate an ergodic observable in the expectation operator, see \textcite{PetersAdamou2018a,AdamouETAL2021}.
In the case of time lotteries, when uncertainty appears in the time domain it enters the growth rate calculation in the denominator. This makes it impossible to find a corresponding mapping using the expected value, see \ssref{Kunstgriff}.
Thus from a theoretical perspective, our treatment of timing uncertainty complements the ergodicity economics literature.

From an experimental perspective, our reanalysis of data in \secref{gDiff} joins recent experimental evidence of strong dependence on wealth dynamics in human decisions under uncertainty \parencite{MederETAL2019} and may be used to design similar experiments outlined in \ssref{DistinguishingExperiments}, which are planned.


\section{Theoretical Framework} \seclabel{model}
We begin by formalising a choice problem between \textit{time lotteries}. A decision maker has to choose between two possible gambles, or time lotteries, which involve fixed payments at either certain or uncertain timing. We follow the formalism introduced in \textcite{DeJarnetteETAL2020} with small modifications to notation.

\begin{definition}[Timed Payment]
A Timed Money Payment \TP is an ordered pair $\left(t, \Dx\right)$, which denotes receiving an amount of money \Dx at time $t$.
\end{definition}

The payment amount and the payment time can be negative in general. However, for concreteness we confine our attention to positive payments ($\Dx > 0$) and only to future payments, which are the most commonly considered dilemmas. Timed payments can now be used to define a time lottery.


\begin{definition}[Time Lottery]
Consider two timed payments $\TP_1$ and $\TP_2$, which correspond to the same payment \Dx at times \tI and \tII and a probability $0 \leq p \leq 1$ to receive the payment at the earlier time $\tI \left(\leq \tII\right)$. Then the tuple $\left(\tI, \tII, p, \Dx\right)$ defines a time lottery \TL. A time lottery provides a decision maker with the option to receive a certain payment \Dx either at \tI (with probability $p$) or the same payment at \tII (with probability $1-p$).

For every time lottery \TL there exists a unique corresponding degenerate time lottery \DTL for which the payment \Dx is received with certainty at the expected payment time, $\EA{t} = p\tI + \left(1-p\right)\tII$. Thus any degenerate time lottery \DTL is a timed payment.
\end{definition}

The setup assumed in the definitions above is presented in \fref{TLSetup}, which illustrates the choice problem between a risky time lottery \TL and its corresponding riskless degenerate time lottery \DTL.

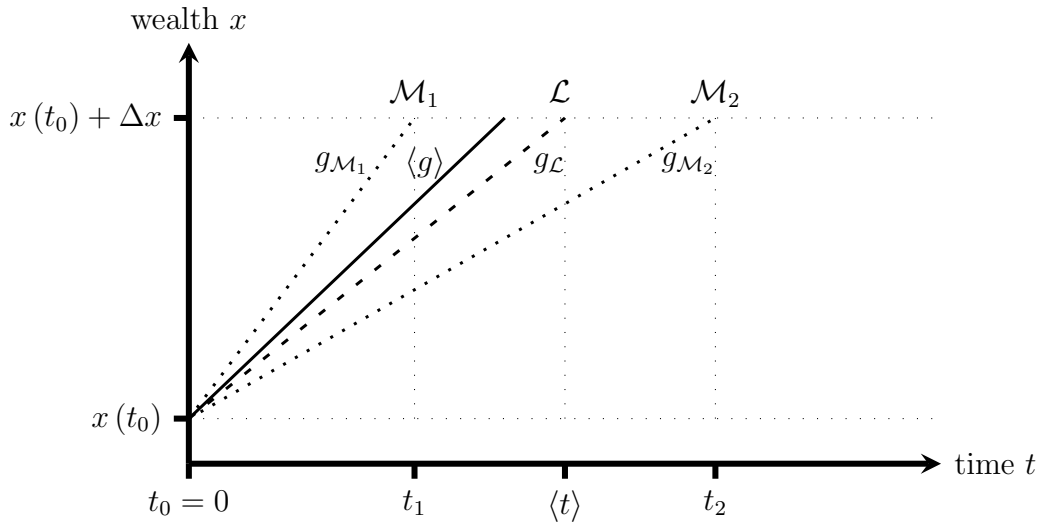
\begin{figure}[htb]
\begin{center}
	\begin{tikzpicture}[
	>=stealth																						
	]
		\draw [dash pattern=on \pgflinewidth off 5pt] (0,0) -- (10,0);
		\draw [dash pattern=on \pgflinewidth off 5pt] (0,4) -- (10,4);
		\draw [->, line width=0.8mm] (0,-0.6) -- (10,-0.6)	node[right,xshift=0pt]	{time $t$};
		\draw [->, line width=0.8mm] (0,-0.6) -- (0,5) 		node[above,yshift=0pt]	{wealth $x$};
		
		\draw[-,line width=0.8mm] (0,-0.8)--(0,-0.6) node[below, yshift=-5pt]{$t_0 = 0$};
		\draw[-,line width=0.8mm] (3,-0.8)--(3,-0.6) node[below, yshift=-5pt]{$\tI$};
		\draw[-,line width=0.8mm] (5,-0.8)--(5,-0.6) node[below, yshift=-5pt]{$\EA{t}$};
		\draw[-,line width=0.8mm] (7,-0.8)--(7,-0.6) node[below, yshift=-5pt]{$\tII$};

		\draw[-,line width=0.8mm] (-0.2,0)--(0,0) node[left, xshift=-5pt]{$x\left(t_0\right)$};
		\draw[-,line width=0.8mm] (-0.2,4)--(0,4) node[left, xshift=-5pt]{$x\left(t_0\right) + \Dx$};
		
		\draw [dash pattern=on \pgflinewidth off 5pt] (3,0) -- (3,4) node [above] at (3,4)	{$\TP_1$};
		\draw [dash pattern=on \pgflinewidth off 5pt] (5,0) -- (5,4) node [above] at (5,4)	{\DTL $\vphantom{\TP_1}$};
		\draw [dash pattern=on \pgflinewidth off 5pt] (7,0) -- (7,4) node [above] at (7,4)	{$\TP_2$};
		
		\coordinate (origin) at (0,0);
		\coordinate (x_t1) at (3,4);
		\coordinate (exp_t) at (5,4);
		\coordinate (x_t2) at (7,4);
		\coordinate (bogus) at (4.2,4);

		\draw [loosely dotted,line width=0.4mm] (origin) -- (x_t1) node[anchor=east,pos=0.85] {\gI};
		\draw [line width=0.4mm]	(origin) -- (bogus)	node[anchor=east,pos=0.85] {\EAgr};
		\draw [loosely dashed,line width=0.4mm] (origin) -- (exp_t) node[anchor=west,pos=.85,xshift=5pt] {\DTLgr};
		\draw [loosely dotted, line width=0.4mm] (origin) -- (x_t2) node[anchor=west,pos=0.85,xshift=5pt] {\gII};
	 \end{tikzpicture}
\caption{The time lottery \TL and its corresponding degenerate time lottery \DTL. In a time lottery \TL a certain payment \Dx is received either at an earlier time \tI with probability $p$ or at a later time \tII with probability $1-p$.
The depicted slopes correspond to the growth rates of the corresponding timed payments.
Receiving the payment at the earlier (later) time \tI (\tII) leads to the growth rate \gI (\gII). The ensemble average of these two growth rates associated with the time lottery leads to the ensemble-average growth rate \EAgr.
In the degenerate time lottery, \DTL, the payment \Dx is received with certainty at $\EA{t} = p\tI + \left(1-p\right)\tII$, which yields the growth rate \DTLgr.}
\flabel{TLSetup}
\end{center}
\end{figure}

Our aim is to provide a theory of how decision makers choose between time lotteries, and we define formal criteria for such choices in the next section. Therefore, we continue to describe the problem setup and the necessary definitions for evaluating time lotteries. An important property of preference relations over time lotteries is whether they are risk-averse or risk-seeking. Following \textcite{DeJarnetteETAL2020}, we define risk-averse behaviour over time lotteries.

\begin{definition}[Risk-Averse over Time Lotteries]

The relation $\succsim$ is Risk-Averse over Time Lotteries (RATL) if for all payments \Dx, for all times \tI, \tII, where $\tI \leq \tII$, and for all probabilities $0 \leq p \leq 1$ then
\be \elabel{RATL}
	\DTL \succsim \TL ~.
\ee
\end{definition}

Similarly, a relation $\succsim$ would be ``Risk Seeking over Time Lotteries (RSTL) or Risk Neutral over Time Lotteries (RNTL) if the above holds with $\precsim$ or $\sim$, respectively'' \parencite[4]{DeJarnetteETAL2020}. Put simply, a decision maker is RATL if her preferences are best described by a RATL preference relation, \ie if she would always prefer a certain timed payment over receiving the same amount at an uncertain time, with the same expected payment time as the certain timed payment. A RSTL decision maker would always prefer the lottery with the uncertain time in such a case.

It is also possible to compare two time lotteries with the same payment amount and the same expected payment time, but different payment time risk, \ie variance of payment time. In such a case, we could consider the time lottery with the larger variance of payment time as more risky. Thus, preference to lotteries with smaller variance of payment time would be a natural extension to the definition of RATL. Growth-optimality will provide a decision criterion general enough to predict such behaviours as well.

\section{Evaluating Time Lotteries} \seclabel{EvaluatingTL}

A formal decision theory has to provide a criterion for choosing between two time lotteries. Here we explore what happens if that criterion is maximisation of the growth rate of wealth, or \textit{growth-optimality}. Under this criterion, a time lottery $\TL_1$ is chosen when it corresponds to a higher growth rate of the decision maker's wealth than another time lottery, $\TL_2$, and \textit{vice versa}. This criterion requires defining properly the growth rate that corresponds to a time lottery. In this section we study two possible definitions for such a growth rate and their implications on risk preferences in time lotteries.

A growth rate, \gr, is defined as the scale parameter of time for an underlying dynamic of wealth.\footnote{Assuming wealth follows some dynamic $x\left(t\right)$, it is possible to define a growth rate if there is a transformation $v\left(x\right)$ such that for any time $t$ and time interval $\Delta t$
\be
g = \frac{\Delta v\left(x\right)}{\Delta t}
\ee
is constant, \ie it depends neither on $t$ nor on $\Delta t$ (where $\Delta v\left(x\right)$ is the change in $v$ between $t$ and $t+\Delta t$). This is explained in detail in \textcite{PetersGell-Mann2016,PetersAdamou2018a}. We note that this assumption is inconsequential to our mathematical analysis.}
In the case of time lotteries, as defined above, it is implicit in the literature that the dynamics are additive. Under additive dynamics the change of wealth \Dx over a time interval, between $t$ and $t+\Dt$, is independent of the wealth level at time $t$. In such a case, the wealth, $x$, follows
\be
	x\left(t+\Dt\right) = x\left(t\right) + \Dx	~,
\ee
and the growth rate of wealth is defined as the rate at which wealth changes over the time interval, or
\be
	g = \frac{x\left(t+\Dt\right) - x\left(t\right)}{\Dt} = \frac{\Dx}{\Dt} ~.
\ee
It follows that a timed payment, \TP, defined by the payment $\Dx$ and the payment time $t = t_0 + \Dt$, where $t_0$ is present time and \Dt is a time interval, corresponds to a unique growth rate $\Dx/\Dt$. For simplicity we assume $t_0=0$, so a timed payment \TP corresponds to a growth rate
\be
	g_{\TP} = \frac{\Dx}{t}  ~.
\ee
A time lottery, \TL, is composed of two timed payments, $\TP_1$ and $\TP_2$, corresponding to growth rates $\gI = \nicefrac{\Dx}{\tI}$ and $\gII = \nicefrac{\Dx}{\tII}$, respectively. \TL also defines its unique degenerate time lottery, \DTL, which is a timed payment that corresponds to the growth rate $\DTLgr = \nicefrac{\Dx}{\EA{t}}$. These growth rates are illustrated in \fref{TLSetup}: the slope of the line that connects $\left(t_0, x\left(t_0\right)\right)$ and $\left(\tI, x\left(t_0\right) + \Dx\right)$ is \gI; the slope of the line that connects $\left(t_0, x\left(t_0\right)\right)$ and $\left(\tII, x\left(t_0\right) + \Dx\right)$ is \gII; the slope of the line that connects $\left(t_0, x\left(t_0\right)\right)$ and $\left(\EA{t}, x\left(t_0\right) + \Dx\right)$ is \DTLgr.

To make a growth-optimal choice between time lotteries we would like to define a single growth rate that describes a time lottery. However, the growth rate of wealth in time lotteries is a random variable
\be
	\gr = 
	\begin{cases}
		\gI		=	\frac{\Dx}{\tI}		\qquad \text{with probability $p$}\\
		\gII	=	\frac{\Dx}{\tII}	\qquad \text{with probability $1-p$}  ~,
	\end{cases}
\ee
so it is necessary to collapse the growth rate into a scalar. The scalar growth rates are then used as the criteria to choose between time lotteries.

In general, there are many ways to define a scalar growth rate for a time lottery. We study two possible approaches, the time approach and the ensemble approach. The rationale is to arrive at a growth rate that corresponds to the decision maker's experience of a time lottery. As it turns out, only the time approach truly leads to such a growth rate. Nevertheless, it is still possible that the ensemble approach is in better accordance with decision makers' choices. In \secref{gDiff} we put that to the test using experiments conducted by \textcite{OnayOnculer2007} and \textcite{DeJarnetteETAL2020}.

\subsection{Time Approach} \sseclabel{TimePerspective}

One approach for collapsing the random growth rate of a time lottery into a scalar is the time approach. In this approach the growth rate is defined as the \textit{time-average} growth rate. Technically, this is achieved by taking the long-time limit of the time average of the growth rate. This quantity corresponds to the growth rate experienced by the decision maker if the lottery is indefinitely repeated over time.

Practically, we assume the lottery is sequentially repeated $T$ times and evaluate the growth rate in the limit $T \to \infty$. We assume that in $n_1$ of the times the early payment was realised, and in $n_2 = T - n_1$ of the times it was the later payment.

We denote the time-average growth rate of a time lottery \TL by \TAgr, given by
\begin{align}
	\elabel{TAGR1}
	\TAgr &\equiv \lim_{T\to\infty} ~ \frac{\text{total payment after $T$ rounds}}{\text{total time elapsed after $T$ rounds}} \\
	\elabel{TAGR2}
	&= \lim_{T\to\infty} ~ \frac{T \Dx}{n_1 \tI + n_2 \tII} = \lim_{T\to\infty} ~ \frac{\Dx}{\nicefrac{n_1}{T}\cdot\tI + \nicefrac{n_2}{T}\cdot\tII} \\
	\elabel{TAGR3}
	&= \frac{\Dx}{p\tI + \left(1-p\right)\tII} = \frac{\Dx}{\EA{t}}	~.
\end{align}

The time-average growth rate, \TAgr, is illustrated by the slope of the dashed line in \fref{TLSetup} and denoted by \DTLgr, the notation will become clear shortly.

\subsection{Ensemble Approach} \sseclabel{EnsemblePerspective}

Another approach for collapsing the random growth rate into a scalar is the ensemble approach. Here, the scalar growth rate is defined as the \textit{ensemble-average} growth rate or, synonymously, as the expected value of the growth rate. Technically, this is achieved by taking the large sample limit of the ensemble average of the growth rate, which simply gives the arithmetic mean of the random growth rate.

Practically, we assume the lottery is simultaneously realised $N$ times and denote by $\gr_i$ the $i$-th realised growth rate. As before, we assume that in $n_1$ realisations the early payment was realised and in $n_2 = N - n_1$ realisations it was the later payment. We then evaluate the mean of these growth rates in the limit $N\to\infty$.

We denote the ensemble-average growth rate of a time lottery \TL by \EAgr, given by 
\begin{align}
	\EAgr &\equiv \lim_{N\to\infty}  \frac1N \sum_{i=1}^N \gr_i = \lim_{N\to\infty}  \frac1N \left( n_1 \cdot \gI + n_2 \cdot \gII \right) \\
	&= \lim_{N\to\infty} \left( \frac{n_1}{N} \cdot \frac{\Dx}{\tI} + \frac{n_2}{N} \cdot \frac{\Dx}{\tII} \right) \\
	\elabel{EAgrTL}
	&= p \frac{\Dx}{\tI} + \left(1-p\right) \frac{\Dx}{\tII} = \EA{\frac{\Dx}{t}} ~.
\end{align}
The ensemble-average growth rate, \EAgr, is illustrated by the slope of the solid line in \fref{TLSetup}. 
This quantity describes the decision maker's experience as if $N\to\infty$ simultaneous lotteries were realised, with their outcomes pooled together and divided by $N$.

We note that a decision maker may never experience the ensemble-average growth rate.
This is also illustrated in \fref{TLSetup} by the difference of the slopes of \EAgr and \DTLgr.
The ensemble-average and time-average growth rates differ in the way the decision maker's experience in time lotteries is quantified. It is therefore no surprise that the ensemble-average growth rate of a time lottery, \EAgr, does not coincide with the time-average growth rate, \TAgr, where the former is generally higher. They will only coincide if $\tI = \tII$ or if $p$ is either 0 or 1, \ie if the time lottery comprises only a single timed payment, \ie if it is, in fact, a degenerate time lottery.
Put differently, there exists no real-world mechanism under the time lottery that would result in the time interval that corresponds to \EAgr.\footnote{Given a time lottery that pays $\Dx$ at either $\tI$ with probability $p$ or at $\tII$ with probability $\left(1-p\right)$, this effective time interval is the time at which the solid line in \fref{TLSetup} crosses the horizontal line $x\left(t_0\right) + \Dx$. It equals $\frac{\tI\tII}{\tI+p\left(\tII - \tI\right)}$. This effective time does not correspond to an actual payment time or an average payment time at which the decision maker receives $\Dx$.}



\subsection{The Failure of the ``Kunstgriff''} \sseclabel{Kunstgriff}

The inequality between the scalar growth rates
\begin{equation} \elabel{ErgodicHypothesis}
  \TAgr\neq\EAgr
\end{equation}
implies that the growth rate of wealth is non-ergodic in time lotteries. In the context of statistical mechanics, \person{Boltzmann} referred to ergodicity, \ie the equality of the time average and the ensemble average of an observable (in our case $\TAgr = \EAgr$) and hence their substitutability, as a mathematical ``Kunstgriff'' (trick). This trick turns out to be helpful in simplifying many calculations \parencite[345]{Boltzmann1872}, and was labelled later the ergodic hypothesis.


Every decision theory that relies on an expected value in its optimand is paired with the ergodicity hypothesis explicitly or implicitly. Under certain conditions, transforming a non-ergodic observable (such as the wealth of a decision maker) into an ergodic observable yields equivalence between expected utility theory (EUT) and the maximisation of the time-average growth rate \parencite{PetersGell-Mann2016,PetersAdamou2018a}. Under the appropriate transformation of wealth, both theories would give identical predictions of decision makers' choices. Furthermore, this equivalence produces predictions of decision makers' utility functions, which can be experimentally tested \parencite{PetersGell-Mann2016,MederETAL2019}.

In the case of time lotteries, such transformation of wealth does not exist. Since the dynamics are additive the appropriate transformation or utility function is the identity, $u(x) = x$ \parencite{PetersGell-Mann2016,PetersAdamou2018a}. Although the appropriate mapping is applied in \eref{EAgrTL}, the ensemble-average and time-average growth rates are not equal. This is a result of introducing the uncertainty in the payment time and not in the payment amount, where the ergodicity transformation could counteract the non-ergodicity. Thus the uncertainty enters \eref{TAGR3} in the denominator, which makes it impossible to construct an ergodic growth rate. Put differently, there exists no function $u$ of wealth $x$, which induces an equality between the ensemble average on the left-hand side and the time average on the right-hand side of the following equation:
\begin{align}
	\elabel{KunstgriffI}
	\EA{\frac{\D u(x)}{t}} &\overset{?}{=} \dfrac{\Dx}{\EA{t}} \\
	p \frac{\D u(x)}{\tI} + \left(1-p\right) \frac{\D u(x)}{\tII} &\overset{?}{=} \dfrac{\Dx\displaystyle}{p\tI + \left(1-p\right)\tII} \\
	\elabel{KunstgriffII}
	\D u(x) &\overset{?}{=} \frac{\tI \tII}{(p\tI + \left(1-p\right)\tII)(p\tII + \left(1-p\right)\tI)} \Dx ~.
\end{align}
For the equality in \eref{KunstgriffI} to be obtained it is therefore necessary for $u(x)$ to also depend on the payment times, \ie on the problem setup, rather than on wealth $x$ and its dynamics.
%


\subsection{Using Growth Rates to Evaluate Time Lotteries} \sseclabel{GrowthOptimality}

\subsubsection{Growth-Optimality}
We will now use the definitions above as choice criteria between time lotteries. Choosing between two time lotteries, $\TL_a$ and $\TL_b$, requires evaluating the growth rate associated with each, \ie $\gr_a$ and $\gr_b$. We define a preference relation of a decision maker who chooses growth-optimally between time lotteries.

\begin{definition}[Growth-Optimal over Time Lotteries] \dlabel{GOTL}
The relation $\succsim$ is Growth-Optimal over Time Lotteries (GOTL) if given two time lotteries, \TLa and \TLb, with growth rates $g_a$ and $g_b$, respectively,
\begin{enumerate}
	\item $\TLa \succ \TLb$ \quad[`\TLa is preferred to \TLb']\quad 				if and only if \quad $\gra > \grb$ 
	\item $\TLa \sim \TLb$ \quad[`indifference between \TLa and \TLb']\quad if and only if \quad $\gra = \grb$
	\item $\TLa \prec \TLb$ \quad[`\TLb is preferred to \TLa']\quad 				if and only if \quad $\gra < \grb$ ~.
\end{enumerate}
\end{definition}

In words, a decision maker is GOTL if she prefers lottery \TLa if her wealth grows faster under this choice than under choice \TLb, and vice versa. While intuitive and straightforward, such preferences are not commonly studied in economics, but they adhere to the standard rationality definitions.

\begin{proposition}[Optimization of Growth satisfies von Neumann-Morgenstern axioms]
	If a relation $\succsim$ is GOTL, it satisfies the axioms by \person{von Neumann-Morgenstern}.
\label{prop:trans}
\end{proposition}

\begin{proof}
See Appendix \ref{app:appA}.
\end{proof}

We can now turn to test what growth-optimal preferences mean for decision makers' choices in time lotteries, based on the two approaches described above. 

\subsubsection{Risk-Neutrality over Time Lotteries in the Time Approach} 

\begin{proposition}[Time Approach Predicts RNTL]
	In the time approach, if a relation $\succsim$ is GOTL, then $\succsim$ is RNTL (risk-neutral over time lotteries).
\label{prop:rntl}
\end{proposition}

\begin{proof}
\sseclabel{TimeApproachIsRNTL}

We would like to show that in the time approach, if a relation $\succsim$ is growth-optimal, then $\succsim$ is RNTL (risk-neutral over time lotteries).

To show that $\succsim$ is RNTL, a decision maker must be indifferent between the time lottery and its degenerate time lottery, $\TL \sim \DTL$.
In the time approach the growth rate of the time lottery \TL is given by \eref{TAGR1}:
\begin{align}
	\TAgr = \frac{\Dx}{p\tI + \left(1-p\right)\tII} = \frac{\Dx}{\EA{t}}	~,
\end{align}

The growth rate of the degenerate time lottery \DTL is $\DTLgr = \nicefrac{\Dx}{\EA{t}}$. Hence, the time-average growth rate of the risky time lottery coincides with the growth rate of its riskless degenerate time lottery. Because $\succsim$ is growth-optimal, $\TL \sim \DTL$.

\end{proof}

Thus, maximising the time-average growth rate predicts decision makers would be indifferent between risky time lotteries and their corresponding riskless degenerate time lotteries.
The ensemble approach yields a different prediction of growth-optimal behaviour.

\subsubsection{Risk-Seeking Behaviour over Time Lotteries in the Ensemble Approach} 

\begin{proposition}[Ensemble Approach Predicts RSTL]
	In the ensemble approach, if a relation $\succsim$ is GOTL, then $\succsim$ is RSTL (risk-seeking over time lotteries).
\label{prop:rstl}
\end{proposition}

\begin{proof}

We would like to show that in the ensemble approach, if a relation $\succsim$ is growth-optimal, then $\succsim$ is RSTL (risk-seeking over time lotteries).

For that purpose we look at a general time lottery \TL, which is composed of two timed payments, $\TP_1$ and $\TP_2$, corresponding to growth rates $\gI = \nicefrac{\Dx}{\tI}$ and $\gII = \nicefrac{\Dx}{\tII}$, respectively. The time lottery \TL also defines its unique degenerate time lottery, \DTL, which is a timed payment that corresponds to the growth rate $\DTLgr = \nicefrac{\Dx}{\EA{t}}$. We also assume that the time lottery \TL is not just a timed payment, \ie that $\tI \neq \tII$ and $p$ is neither 0 nor 1. To show that $\succsim$ is RSTL, \TL must be preferred to \DTL.

In the ensemble approach the growth rate of \TL is given by \eref{EAgrTL}:
\be
	\EAgr = \EA{\frac{\Dx}{t}} = p \frac{\Dx}{\tI} + \left(1-p\right) \frac{\Dx}{\tII}  ~.
\ee
We define $f\left(z\right) = \nicefrac{1}{z}$, thus $f$ is a convex function and from \person{Jensen's} inequality follows
\begin{align}
  f\big(\EA{t}\big) \quad &< \quad \EA{f(t) \vphantom{^2} }	\\
  \elabel{Jensen1}
  \Dx \frac{1}{\EA{t}} \quad &< \quad \Dx \EA{\frac1t}		\\
  \frac{\Dx}{\EA{t}} \quad &< \quad p \frac{\Dx}{\tI} + \left(1-p\right) \frac{\Dx}{\tII}	\\
  \elabel{TAlowerEA}
  \DTLgr \quad &< \quad \EAgr	~.
\end{align}
Because $\succsim$ is growth-optimal, the time lottery is preferred over the corresponding degenerate time lottery, $\TL \succ \DTL$.

\end{proof}

Thus, maximising the ensemble-average growth rate means that decision makers would prefer risky time lotteries over their corresponding riskless degenerate time lotteries.\footnote{This is unless $g_a = g_b$. In the ensemble approach this happens only if $\tI=\tII$, or if $p \in \{0,1\}$, \ie the time lottery is a timed payment.} This is mathematically equivalent to maximising expected discounted utility, $\EV{D(t)\Delta u(x))}$, for linear utility ($u\left(x\right)=x$) and hyperbolic discounting ($D(t)=\nicefrac1t$). As described in \textcite{PetersAdamou2018a,AdamouETAL2021}, these are indeed the appropriate functions for the specified problem and dynamics. Thus, this coincides with the prediction of EDUT that expected discounted utility maximisers must be RSTL for any utility function and any convex discount function.

It is worth repeating the main argument from \ssref{Kunstgriff} about the failure of the ergodicity trick here.
There exists no concave utility function $u$ which if inserted in the right-hand side of \eref{Jensen1} induces an equality.

The ensemble and time approaches not only result in different growth rates but also lead to different risk preferences in time lotteries.
We can now put the two decision criteria -- the ensemble-average growth rate and the time-average growth rate -- to the test by revisiting experiments conducted by \textcite{OnayOnculer2007} and \textcite{DeJarnetteETAL2020}. We would like to find out which of the two, if any, is in accordance with the experimental evidence.

\section{Revisiting Experimental Evidence on Time Lotteries} \seclabel{gDiff}

The growth-optimality criterion makes two predictions on risk preferences in time lotteries, depending on how growth rates are computed. We have just established that in the ensemble approach decision makers must be risk-seeking over time lotteries. Practically, when facing a choice between a time lottery and its corresponding degenerate time lottery, the ensemble approach predicts a preference for the risky lottery. In the case of a choice between two time lotteries with the same payment amount and the same expected payment time, but a different payment time risk, decision makers must prefer the lottery with the higher variance in payment time.

The time approach predicts decision makers are risk-neutral over time lotteries. This means decision makers must be indifferent between a time lottery and its corresponding degenerate time lottery, and between lotteries with the same payment amount and the same expected payment time. If decision makers are forced to choose, the time approach does not provide a prediction for this choice. It may be hypothesised that indifference means that subjects are neither significantly risk-seeking nor risk-averse. Yet, it is possible that there are other higher-order effects that would systematically affect subjects' choices and are not captured by our model of choice alone.

We put these predictions to the test using evidence from two experiments \parencite{OnayOnculer2007,DeJarnetteETAL2020}. For each experiment we calculate the corresponding growth rate of every time lottery under the ensemble and time approaches, using \eref{EAgrTL} in the ensemble approach and \eref{TAGR3} in the time approach.

\Tref{experimentDeJarnette} presents the results for Part I of the experiment in \textcite{DeJarnetteETAL2020}. We reproduce Table 2 of their Appendix D adding the corresponding growth rates. In this experiment a total number of 196 subjects had to choose between two time lotteries in ten different settings.\footnote{Each subject had to make only five choices in total and received payment for only one of the five choices, selected randomly. For more details on the experiment please refer to Appendix D in \textcite{DeJarnetteETAL2020}.} In six settings subjects had to choose between a time lottery and its corresponding degenerate time lottery (questions 1--3). In the other four settings the choice was between two time lotteries with the same payment amount and the same expected payment time, but a different payment time risk. In \tref{experimentDeJarnette} option \oI refers to the less risky option (or the degenerate time lottery) and option \oII to the riskier option.


\begin{table}[!htb]
\ra{1.25}
\scriptsize
\centering
\captionof{table}{Growth rates for the experimental results in \textcite{DeJarnetteETAL2020} (in $\$/wk$).}\tlabel{experimentDeJarnette}
\begin{tabular}{lccccc}
\addlinespace
\toprule[1.5pt]
\addlinespace
Question & $\EAgr^{\oI}$ & $\EAgr^{\oII}$ & $\TAgr^{\oI}$ ($=\TAgr^{\oII}$) & $\EAgr^{\oII}-\TAgr^{\oII}$ & \makecell{Fraction of RATL\\subjects (\%)} \\
\midrule[1.5pt]
Question 1, long treatment 			& 10.0 & 16.0 & 10.0 & 6.0	& 65.7  \\
Question 1, short treatment			& 10.0 & 16.0 & 10.0 & 6.0	& 56.0  \\
Question 2, long treatment 			&  5.0 &  6.9 &  5.0 & 1.9	& 50.5  \\
Question 2, short treatment			&  5.0 &  9.0 &  5.0 & 4.0	& 55.0  \\
Question 3, long treatment 			&  5.0 &  6.7 &  5.0 & 1.7	& 48.6  \\
Question 3, short treatment			&  5.0 &  6.7 &  5.0 & 1.7	& 37.4  \\
Question 4, long treatment 			&  8.3 & 12.5 &  8.0 & 4.5	& 64.8  \\
Question 4, short treatment			&  8.3 &  8.8 &  8.0 & 0.8 & 38.5  \\
Question 5, long treatment 			&  5.3 & 11.6 &  4.3 & 7.3	& 73.3  \\
Question 5, short treatment			&  3.5 &  3.0 &  2.9 & 0.1	& 52.8  \\
\bottomrule[1.5pt]
\end{tabular}

\flushleft{\textit{Notes:} In each question subjects had to choose between two options, labelled \oI and \oII. $\EAgr^{\oI}$ ($\EAgr^{\oII}$) refers to the ensemble-average growth rate associated with option \oI (\oII). Similarly $\TAgr^{\oI}$ ($\TAgr^{\oII}$) refers to the time-average growth rate associated with option \oI (\oII). The payment amount and the expected payment time were similar for the two options, so $\TAgr^{\oI} =\TAgr^{\oII}$. In questions 1--3 option \oI was the degenerate time lottery corresponding to option \oII in the respective question. In these cases $\EAgr^{\oI}=\TAgr^{\oI}$, since option \oI was a timed payment. In questions 4 and 5, options \oI and \oII were two non-degenerate time lotteries with the same payment amount and the same expected payment time, while in option \oI the payment time had lower variance than in option \oII. Thus, in all questions option \oI was the less risky time lottery and option \oII the more risky. The fraction of RATL subjects is the fraction of subjects who preferred option \oII to option \oI in each question. The long and short treatments refer to two different experiments with a similar structure but with different parameters, \ie different payment amounts, payment times and probabilities. Questions 1 and 3 were similar in both treatments. For more details on the experiments please refer to Appendix D in \textcite{DeJarnetteETAL2020}.}
\end{table}

\Tref{experimentOnay} shows similar results for Study 1 in \textcite{OnayOnculer2007}. In this study 55 subjects had to choose between time lotteries and their corresponding degenerate time lottery in six different settings. Unlike the experiments in \textcite{DeJarnetteETAL2020}, this experiment was not consequential, and participants only received a flat rate payment for their participation.

\begin{table}[!htb]
\ra{1.25}
\scriptsize
\centering
\captionof{table}{Growth rates for the experimental results in \textcite[Tables 2 \& 5]{OnayOnculer2007} (in $NTL/mth$).}\tlabel{experimentOnay}
\begin{tabular}{lcccccc}
\addlinespace
\toprule[1.5pt]
\addlinespace
Case & \EA{t} ($mths$) & \Dx ($NTL$) & \EAgr & \TAgr & $\EAgr-\TAgr$ & \makecell{Fraction of RATL\\subjects (\%)} \\
\midrule[1.5pt]
1	& 9	& 160	&	43.6	&	27.8	&	25.9	&	22\\
2	& 9	& 140	&	38.2 	&	15.6	&	22.6	&	9\\
3	& 6	& 160	&	87.3	&	26.7	&	60.6	&	62\\
4	& 6	& 140	&	76.4	&	23.3	&	53.0	&	40\\
5	& 2	& 160	&	145.5	&	80.0	&	65.5	&	75\\
6	& 2	& 140	&	127.3	&	70.0	&	57.3	&	93\\
\bottomrule[1.5pt]
\end{tabular}

\flushleft{\textit{Notes:} In all cases subjects had to choose between a time lottery and its corresponding degenerate time lottery. In cases 2, 4 and 6, the question presented to the subjects was framed in terms of a loss of 140 NTL rather than a gain. For comparability with the other experiments we consider absolute amounts. The fraction of RATL subjects in cases 2, 4 and 6 is therefore found by subtracting the ``proportion of subjects that behaved as predicted by DEU'' in \textcite[Tab. 2]{OnayOnculer2007} from $100\%$. For more details on the experiments please refer to Section 2 in \textcite{OnayOnculer2007}.}
\end{table}

In each table and for every choice problem we include the growth rate associated with each option using the ensemble approach (\EAgr) and the time approach (\TAgr). We also include the difference between these growth rates and the fraction of subjects who were RATL, thus violating the EDUT prediction. We note that the growth rate magnitudes are not comparable between the experiments. The growth rates in the tables are given in the relevant units for each of the experiments, \ie US Dollars per week (${\$}/wk$) in \tref{experimentDeJarnette} and New Turkish Lira per month ($NTL/mth$) in \tref{experimentOnay}.

These results demonstrate that the ensemble-average growth rate is always higher than the time-average growth rate, $\EAgr > \TAgr$, a result of \person{Jensen's} inequality (see also the proof of \propref{rstl}). It is also illustrated in \fref{TLSetup}, where \EAgr corresponds to the slope of the solid line, which is higher than the time-average growth rate \TAgr, the slope of the line that connects the points $\left(t_0, x\left(t_0\right)\right)$ and $\left(\EA{t}, x\left(t_0\right) + \Dx\right)$. This creates the illusion of faster growth, while in practice, \EAgr is inaccessible as has been clearly identified in the behavioural ecology \textcite[see][p. 16, Box 2.1]{StephensKrebs1986}. For degenerate time lotteries there is no difference between the two growth rates.

The experimental results are not straightforward to interpret, and they do not provide a clear-cut answer as to whether decision makers are generally risk-seeking, risk-neutral or risk-averse in time lotteries. In both experiments, many subjects were risk-seeking and many were risk-averse.
Nevertheless, the difference between the ensemble-average and time-average growth rates reveals a striking regularity. In both experiments, as the ensemble-average growth rate gets larger with respect to the time-average growth rate, decision makers are more likely to prefer less risky lotteries, \ie the bigger the difference is, the more subjects were RATL. This regularity is seemingly counterintuitive. Under growth-optimality, when the growth rate associated with a particular choice is very high, one would expect it to be attractive to decision makers. But the higher the ensemble-average growth rate is, relative to the sure payment, the less attractive it becomes. This is illustrated in \fref{RATL_vs_growthrates_diff}.

\begin{figure}[!htb]
\centering
\includegraphics[width=1.0\textwidth]{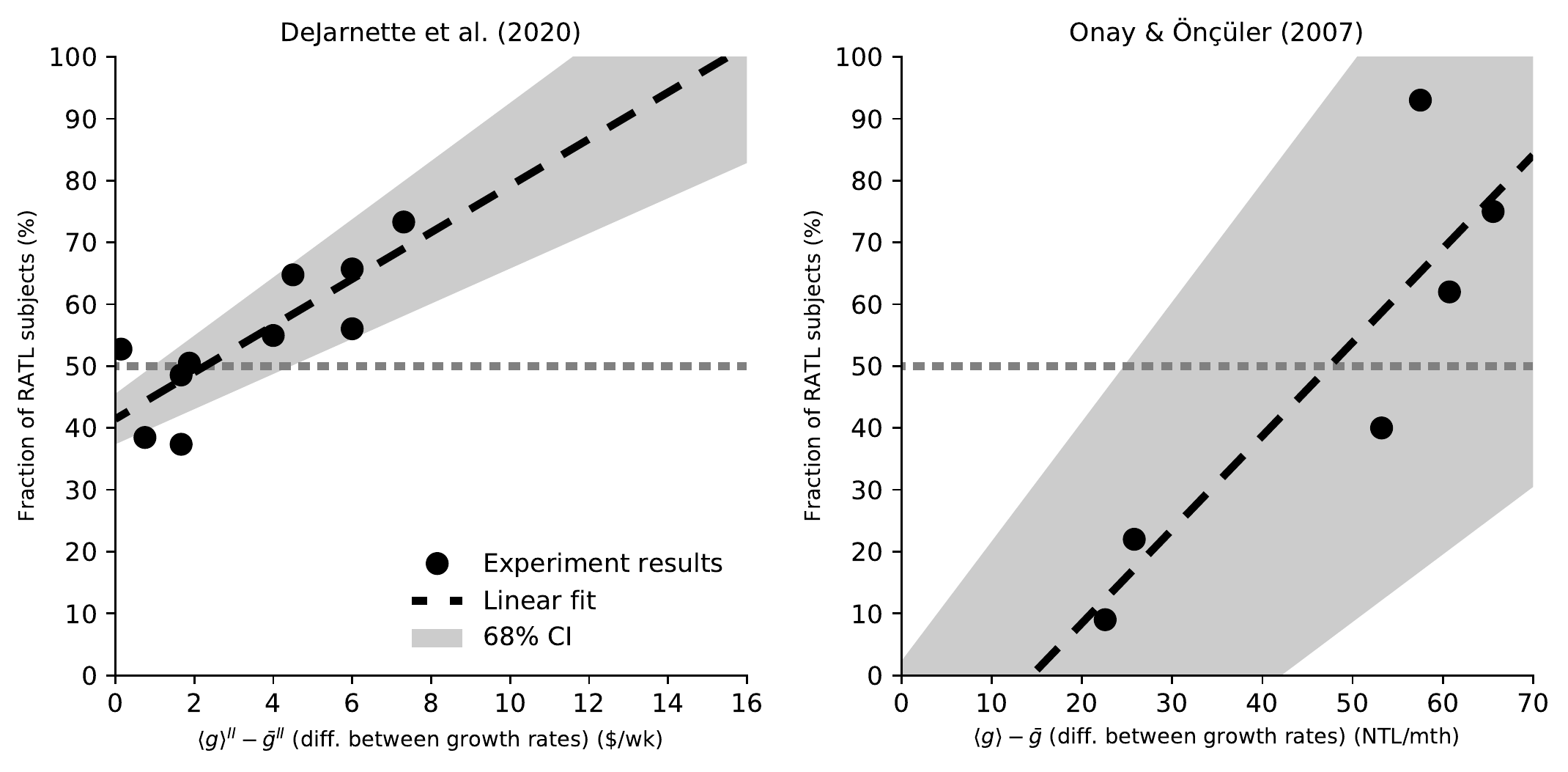}
\caption{The relationship between the fraction of RATL subjects and the difference between the ensemble-average and time-average growth rates ($\EAgr-\TAgr$) in incentivised experiments (see \tref{experimentDeJarnette} and \tref{experimentOnay}). Left) Results for Part I of the experiment in \textcite{DeJarnetteETAL2020}. Right) Results for Study 1 in \textcite{OnayOnculer2007}. The black lines are OLS linear fits for the fraction of RATL subjects as a function of $\EAgr-\TAgr$ ($R^2 = 0.67$ in \textcite{DeJarnetteETAL2020}, $R^2 = 0.76$ in \textcite{OnayOnculer2007}). The grey areas represent 68\% (1$\sigma$) confidence intervals based on the linear fits.}
\flabel{RATL_vs_growthrates_diff}
\end{figure}

Thus, the experimental results falsify the ensemble approach, or EDUT. Strictly speaking, EDUT predicts that all subjects must be risk-seeking. More realistically, it would predict an inverse relationship between the fraction of RATL subjects and the difference $\EAgr-\TAgr$, since the riskier option becomes more attractive if evaluated using its ensemble-average growth rate. However, the relationship found in the experiments is significantly positive.


In the experiment conducted by \textcite{DeJarnetteETAL2020} we find that if the difference between the ensemble-average and time-average growth rates is small, decision makers are neither significantly risk-seeking nor risk-averse. Thus, in such cases, the prediction of the time approach, that decision makers will be risk-neutral in time lotteries, is in line with the experimental evidence. This indicates that the time approach may provide a better model of decision-making, also echoing the experimental results in \textcite{MederETAL2019}.

The risk-neutrality predicted in the time approach may practically lead to arbitrary fractions of RATL. Thus, the time approach predicts neither an upward nor a downward sloping relationship for the fraction of RATL subjects and the difference $\EAgr-\TAgr$, and the experimental results do not support or falsify it. Other effects, possibly systematic, may affect decision makers' choices. It is also possible that if decision makers are forced to choose, there are higher-order decision criteria that are not captured by growth-optimality alone. Other experiments, with lotteries for which the time approach does not predict indifference, are therefore necessary to confirm or falsify it.


\subsection{Distinguishing Experimental Design}
\sseclabel{DistinguishingExperiments}
In our basic setup, decision makers had to choose between a risky time lottery and its corresponding riskless degenerate time lottery. This is a choice between two options with the same payment amount and the same expected payment time. In such a case, maximising the time-average growth rate predicts risk-neutrality, \ie indifference between the risky time lottery and the riskless degenerate time lottery. To design experiments which may falsify or confirm the time approach prediction the choice has to be between lotteries that differ either in the payment amount or in the expected payment time.

\paragraph{Setup 1 -- adjusting times}
A simple setup of such a choice problem would be between a risky time lottery, \TL, and a riskless timed payment, \TP, with the same payment amount, \Dx and is depicted in \fref{AdjustingTimes}. The expected payment time of the time lottery, $\EA{t}$, would now differ from the payment time of the timed payment, $t_{\TP}$.
If $\EA{t} > t_{\TP}$ the time approach predicts $\TP \succ \TL$, since
\be
\TAgr_{\TP} = \nicefrac{\Dx}{t_{\TP}} > \nicefrac{\Dx}{\EA{t}} = \TAgr_{\TL} ~,
\ee
and $\TP \prec \TL$ if $\EA{t} < t_{\TP}$.
The advantage of this setup is the invariance to subjective perception of money.

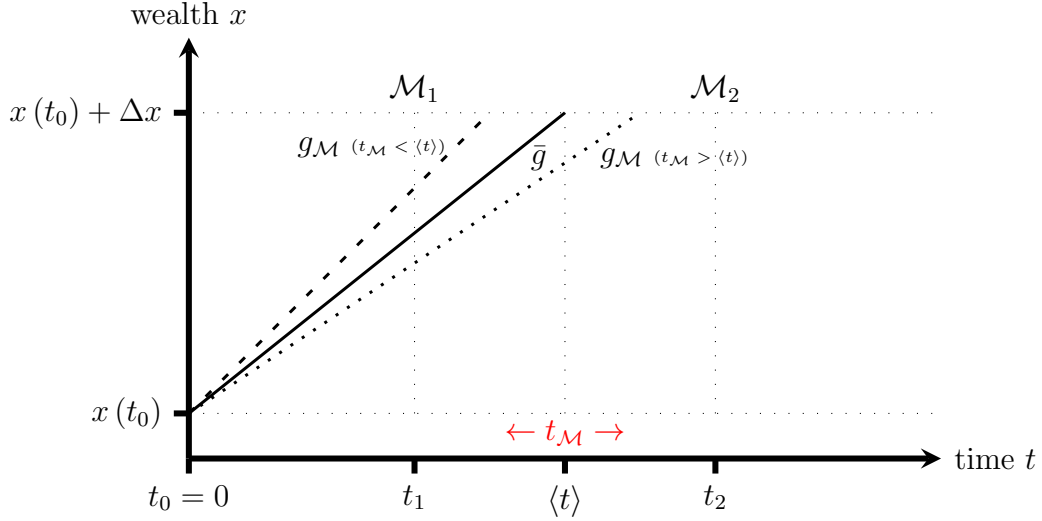
\begin{figure}[h]
\begin{center}
  \begin{tikzpicture}[
	>=stealth																						
	]
		\draw [dash pattern=on \pgflinewidth off 5pt] (0,0) -- (10,0);
		\draw [dash pattern=on \pgflinewidth off 5pt] (0,4) -- (10,4);
		\draw [->, line width=0.8mm] (0,-0.6) -- (10,-0.6)	node[right,xshift=0pt]	{time $t$};
		\draw [->, line width=0.8mm] (0,-0.6) -- (0,5) 		node[above,yshift=0pt]	{wealth $x$};
		
		\draw[-,line width=0.8mm] (0,-0.8)--(0,-0.6) node[below, yshift=-5pt]{$t_0 = 0$};
		\draw[-,line width=0.8mm] (3,-0.8)--(3,-0.6) node[below, yshift=-5pt]{$\tI$};
		\draw[-,line width=0.8mm] (5,-0.8)--(5,-0.6) node[below, yshift=-5pt]{\EA{t}};
		\draw[-,line width=0.8mm] (5,-0.6) node[above, yshift=0pt]{\red{$\leftarrow t_{\TP} \to$}};
		\draw[-,line width=0.8mm] (7,-0.8)--(7,-0.6) node[below, yshift=-5pt]{$\tII$};

		\draw[-,line width=0.8mm] (-0.2,0)--(0,0) node[left, xshift=-5pt]{$x\left(t_0\right)$};
		\draw[-,line width=0.8mm] (-0.2,4)--(0,4) node[left, xshift=-5pt]{$x\left(t_0\right) + \Dx$};
		
		\draw [dash pattern=on \pgflinewidth off 5pt] (3,0) -- (3,4) node [above] at (3,4)	{$\TP_1$};
		\draw [dash pattern=on \pgflinewidth off 5pt] (5,0) -- (5,4) node [above] at (5,4)	{};
		\draw [dash pattern=on \pgflinewidth off 5pt] (7,0) -- (7,4) node [above] at (7,4)	{$\TP_2$};
		
		\coordinate (origin) at (0,0);
		\coordinate (x_t1) at (3,4);
		\coordinate (earlier) at (4,4);
		\coordinate (exp_t) at (5,4);
		\coordinate (later) at (6,4);
		\coordinate (x_t2) at (7,4);
		\coordinate (bogus) at (4.2,4);

		\draw [line width=0.4mm] (origin) -- (exp_t) node[anchor=west,pos=.85,xshift=4pt] {\TAgr};
		\draw [loosely dashed, line width=0.4mm] (origin) -- (earlier) node[anchor=east,pos=.85,xshift=5pt,yshift=4pt] {$g_\TP$ {\tiny ($t_\TP < \EA{t}$)}};
		\draw [loosely dotted, line width=0.4mm] (origin) -- (later) node[anchor=west,pos=.85,xshift=6pt] {$g_\TP$ {\tiny ($t_\TP > \EA{t}$)}};
\end{tikzpicture}
\caption{Setup 1 -- adjusting the payment times. If the payment time $t_{\TP}$ of the riskless timed payment is earlier (later) than the expected payment time, then $t_{\TP} < \EA{t}$ ($t_{\TP} > \EA{t}$) and the corresponding growth rate is higher (lower) than $\TAgr_\TL$. The time approach predicts in this case that decision makers are RATL (RSTL).}
\flabel{AdjustingTimes}
\end{center}
\end{figure}

\paragraph{Setup 2 -- adjusting amounts}
A different distinguishing setup would be a choice problem between a risky time lottery, \TL, and a riskless timed payment, \TP, with payment amounts $\Dx_{\TL}$ and $\Dx_{\TP}$, respectively, depicted in \fref{AdjustingAmounts}.
The expected payment time of the time lottery is the same as the payment time of the timed payment, $\EA{t} = t_{\TP}$. If $\Dx_{\TP} > \Dx_{\TL}$ the prediction is $\TP \succ \TL$, since
\be
\TAgr_{\TP} = \nicefrac{\Dx_{\TP}}{t_{\TP}} > \nicefrac{\Dx_{\TL}}{\EA{t}} = \TAgr_{\TL} ~,
\ee
and $\TP \prec \TL$ if $\Dx_{\TP} < \Dx_{\TL}$.

We note that in EDUT there is no clear-cut prediction of choices in these setups as EDUT is underdetermined. In the first setup the prediction would require specifying a subjective discount function. In the second setup it would require specifying a subjective utility function. If we use the ensemble-average growth rate as the equivalent choice criterion of EDUT instead, we can achieve a clear-cut prediction. It is possible to design setups in which using the time-average growth rate and the ensemble-average growth rate as choice criteria will lead to opposite predictions. For example, in the first setup, if the possible payment times in \TL are \tI and \tII, with probability $p$ to receive the payment at \tI such that
\be
\EA{t} > t_{\TP} > \left(\nicefrac{p}{\tI} + \nicefrac{\left(1-p\right)}{\tII}\right)^{-1}
\ee
then
\be
\TAgr_{\TL} < \TAgr_{\TP} = \EAgr_{\TP} < \EAgr_{\TL} ~,
\ee
so maximising the time-average growth rate predicts preference for the riskless timed payment whereas EDUT predicts preference for the risky time lottery.

\begin{figure}[htb]
\begin{center}
\begin{tikzpicture}[
	>=stealth																						
	]
		\draw [dash pattern=on \pgflinewidth off 5pt] (0,0) -- (10,0);
		\draw [dash pattern=on \pgflinewidth off 5pt] (0,4) -- (10,4);
		\draw [->, line width=0.8mm] (0,-0.6) -- (10,-0.6)	node[right,xshift=0pt]	{time $t$};
		\draw [->, line width=0.8mm] (0,-0.6) -- (0,5) 		node[above,yshift=0pt]	{wealth $x$};
		
		\draw[-,line width=0.8mm] (0,-0.8)--(0,-0.6) node[below, yshift=-5pt]{$t_0 = 0$};
		\draw[-,line width=0.8mm] (3,-0.8)--(3,-0.6) node[below, yshift=-5pt]{$\tI$};
		\draw[-,line width=0.8mm] (5,-0.8)--(5,-0.6) node[below, yshift=-5pt]{$\EA{t}$};
		\draw[-,line width=0.8mm] (7,-0.8)--(7,-0.6) node[below, yshift=-5pt]{$\tII$};

		\draw[-,line width=0.8mm] (-0.2,0)--(0,0) node[left, xshift=-5pt]{$x\left(t_0\right)$};
		\draw[-,line width=0.8mm] (-0.2,4)--(0,4) node[left, xshift=-5pt]{$x\left(t_0\right) + \red{\Dx}$};
		\draw (-.5,4.2) node[above, red]{$\uparrow$};
		\draw (-.5,3.8) node[below, red]{$\downarrow$};
		
		\draw [dash pattern=on \pgflinewidth off 5pt] (3,0) -- (3,4) node [above] at (3,4)	{$\TP_1$};
		\draw [dash pattern=on \pgflinewidth off 5pt] (5,0) -- (5,4) node [above] at (5,4)	{};
		\draw [dash pattern=on \pgflinewidth off 5pt] (7,0) -- (7,4) node [above] at (7,4)	{$\TP_2$};
		
		\coordinate (origin) at (0,0);
		\coordinate (x_t1) at (3,4);
		\coordinate (lower) at (5,3);
		\coordinate (exp_t) at (5,4);
		\coordinate (higher) at (5,5);
		\coordinate (x_t2) at (7,4);
		\coordinate (bogus) at (4.2,4);

		\draw [loosely dashed, line width=0.4mm] (origin) -- (lower) node[anchor=north west,pos=1,yshift=9pt] {$\TAgr_{\TP}$ {\tiny ($\Dx_{\TP} < \Dx_{\TL}$)}};
		\draw [line width=0.4mm] (origin) -- (exp_t) node[anchor=north west,pos=1,yshift=9pt] {$\TAgr$};
		\draw [loosely dotted, line width=0.4mm] (origin) -- (higher) node[anchor=south west,pos=1,yshift=-7pt] {$\TAgr_{\TP}$ {\tiny ($\Dx_{\TP} > \Dx_{\TL}$)}};
\end{tikzpicture}
\end{center}
\caption{Setup 2 -- adjusting the payment amounts. If the payment amount $\Dx_{\TP}$ of the riskless timed payment is higher (lower) than the amount of the risky time lottery, then the corresponding growth rate is higher (lower) than $\TAgr_\TL$, and the time approach predicts that decision makers are RATL (RSTL).}
\flabel{AdjustingAmounts}
\end{figure}
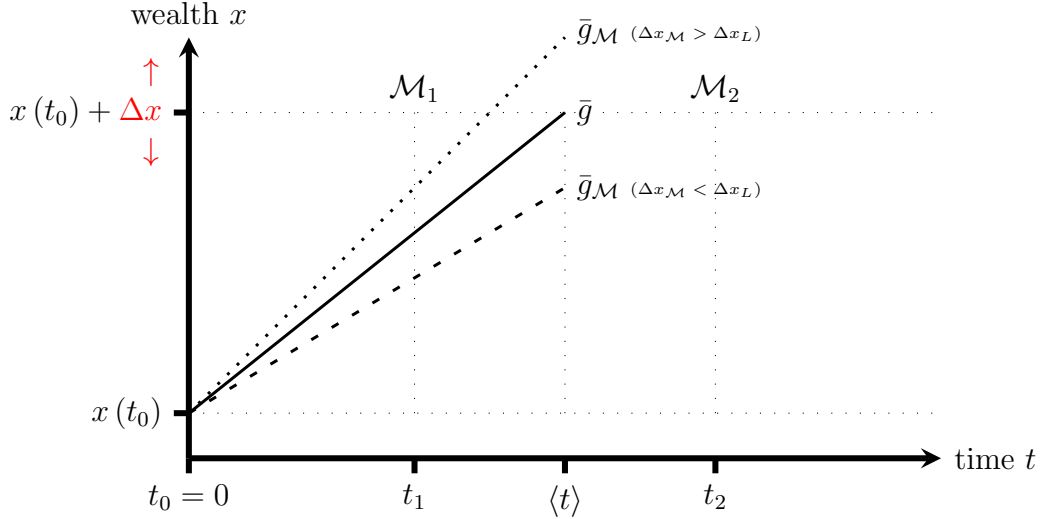

\section{Discussion} \seclabel{Discussion}

This paper explores time lotteries under the postulate that decision makers maximise the growth rate of their wealth. We consider a basic choice problem between a certain payment at a certain time, to the same payment but with an uncertain payment time, assuming the expected payment time is equal to the certain one.

In time lotteries, the growth rate of wealth is a random variable and to make a growth-optimal choice between two lotteries requires collapsing this random variable into a scalar. We present two approaches to compute this scalar. In the first approach, the time approach, the growth rate associated with a time lottery is the growth rate experienced by the decision maker if the lottery is indefinitely repeated over time. Maximisation of the time-average growth rate predicts risk-neutral behaviour in time lotteries.

In the second approach, the ensemble approach, the growth rate associated with a time lottery is the mean or ensemble average of the growth rates, as if the lottery was simultaneously realised infinitely many times. This growth rate does not correspond to any real-world experience of a decision maker. Maximising the ensemble-average growth rate is equivalent to maximising the expected change in discounted utility, and predicts risk-seeking behaviour in time lotteries.
Both approaches are consistent with standard axioms of choice and assume neither behavioural bias nor dynamic inconsistency in the decision maker's behaviour. At all times she prefers the option with the highest growth rate.

Of the two approaches, only the time-average growth rate is a physically meaningful decision criterion, as it corresponds to a decision maker's experience if the time lottery is indefinitely repeated. Together with the empirical evidence, our analysis thus invalidates the decision criterion put forth by the ensemble approach, or EDUT. The time approach predicts risk-neutrality in time lotteries.
Empirically testing a theory where its prediction is risk-neutrality is challenging, as systematic effects not captured by the choice model may influence decisions. This complicates the validation of the time approach. It is not a weakness of this particular approach, but merely a statement that the existing experiments were not designed to test the time approach. Therefore we propose two setups for distinguishing experiments. Applying the core idea of ergodicity economics, using the decision maker's experience over time as a decision criterion, still helps reconciling decision-theoretical reasoning with the empirical evidence.

The main contribution of this paper is the rationalisation of risk-seeking, risk-neutral and risk-averse behaviour in time lotteries in a normative model with a single choice criterion. This allows revisiting existing experimental evidence on risk preferences in time lotteries in light of growth-optimality. We find that the time approach accords better with the experimental evidence than the ensemble approach. We also find, seemingly surprisingly in the context of EDUT, that the higher the ensemble-average growth rate is, relative to the sure payment, the less attractive it becomes. It demonstrates that decision makers may not consider the ensemble-average growth rate as a relevant criterion for their choices. This may be due to the lack of realism in this growth rate, \ie it does not correspond to the real-life experience of decision makers in time lotteries.

Moreover, this paper provides a framework for future experiments. It is possible to design choice problems between time lotteries for which one lottery will be preferred in the ensemble approach and the other in the time approach. Such experiments are able to confirm or falsify the two approaches, and are planned.

We end with a conceptual remark.
Whether or not decision makers' choices maximise the growth rate of their wealth is an open question. Yet, if they do, there is also the question of which growth rate they have in mind when making their choices? Here we present two options.
The ensemble approach uses the expected value, which effectively averages over an ensemble of ``copies'' of the decision maker that cannot be accessed. Moreover, the growth rates that characterise the inaccessible copies cannot be pooled to yield a single realistic growth rate. The ensemble-average growth rate is thus a quantity that never captures the decision maker's experience.
The time approach uses the long time limit, which may also be incompatible with the decision maker's choices if their time horizon is realistically short.
We let data decide which approach better describes how decision makers choose. In the case of time lotteries the time approach seems to be better aligned with the experimental evidence than the ensemble approach. This may reflect the view of \textcite[60]{Kacelnik1997} that the same ``process used for one-off events seems to obey a law that evolved as an adaptation to cope with repetitive events.''


\printbibliography
	
\appendix
\section{Proof of \propref{trans}}
\label{app:appA}


We would like to show that if a relation $\succsim$ is growth-optimal over time lotteries, then it satisfies the \person{von Neumann-Morgenstern} axioms: completeness, transitivity, continuity and independence. For completeness and transitivity there is no difference between the ensemble and time approaches, and the proof is similar. For continuity and independence we separate the proofs for each of the approaches. The ensemble approach is fully consistent with the independence axiom. In the time approach, growth-optimal preferences satisfy independence for time lotteries with the same payment (as commonly described). If time lotteries with different payments are allowed to be compared, this axiom may be violated, and we describe a counterexample.

\subsection*{Completeness}

Any two time lotteries \TLa and \TLb correspond to growth rates \gra and \grb, respectively. It follows that
\begin{align}
\TLa \succ \TLb	&\iff \gra > \grb	 	~, \\
\TLa \sim	\TLb	&\iff \gra = \grb		~, \\
\TLa \prec \TLb	&\iff \gra < \grb		~,
\end{align}
so either \TLa is preferred, \TLb is preferred, or the decision maker is indifferent.

\qed

\subsection*{Transitivity}

For any three lotteries \TLa, \TLb and \TLc, corresponding to growth rates \gra, \grb and \grc, such that
\begin{align}
\TLa \prec \TLb	&\iff \gra < \grb	 	~, \\
\TLb \prec \TLc	&\iff \grb < \grc		~,
\end{align}
if follows that $\gra < \grc$, so $\TLa \prec \TLc$.

Similarly, if
\begin{align}
\TLa \sim \TLb	&\iff \gra = \grb	 	~, \\
\TLb \sim \TLc	&\iff \grb = \grc		~,
\end{align}
then also $\gra = \grc$ and $\TLa \sim \TLc$.

\qed

\subsection*{Continuity}

Given three time lotteries \TLa, \TLb and \TLc, such that $\TLa \prec \TLb$ and $\TLb \prec \TLc$, we wish to show that there exists a weight $\natop\in\left[0,1\right]$ such that $\natop \TLa + \left(1-\natop\right)\TLc \sim \TLb$.

To show that we first need to define the ``natural'' operation \parencite[as described in][]{vonNeumannMorgenstern1944} between the lotteries, \ie, a combination of two time lotteries with a weight \natop. We assume \TLa corresponds to the tuple $\left(t_{a,1}, t_{a,2}, p_a, \Dx_a\right)$ (see the definition of a time lottery in \secref{model}), \TLb to the tuple $\left(t_{b,1}, t_{b,2}, p_b, \Dx_b\right)$ and \TLc to the tuple $\left(t_{c,1}, t_{c,2}, p_c, \Dx_c\right)$. The combination $\natop \TLa + \left(1-\natop\right)\TLc$ is defined as a lottery that pays:
\begin{align}
	\Dx_a\text{ at time } t_{a,1}&\text{, with probability } \natop\cdot p_a  	~, \\
	\Dx_a\text{ at time } t_{a,2}&\text{, with probability } \natop\cdot \left(1-p_a\right)  	~, \\
	\Dx_c\text{ at time } t_{c,1}&\text{, with probability } \left(1-\natop\right)\cdot p_c  	~, \\
	\Dx_c\text{ at time } t_{c,2}&\text{, with probability } \left(1-\natop\right)\cdot \left(1-p_c\right)  	~.
\end{align}
%
%

\paragraph{Continuity in the ensemble approach}

We begin with the ensemble approach. In the ensemble approach the time lotteries \TLa, \TLb and \TLc correspond to the following respective growth rates
\begin{align}
	\EAgra		&=	p_a \frac{\Dx_a}{t_{a,1}} + \left(1-p_a\right) \frac{\Dx_a}{t_{a,2}} ~,\\
	\EAgrb		&=	p_b \frac{\Dx_b}{t_{b,1}} + \left(1-p_b\right) \frac{\Dx_b}{t_{b,2}} ~,\\
	\EAgrc	&=	p_c \frac{\Dx_c}{t_{c,1}} + \left(1-p_c\right) \frac{\Dx_c}{t_{c,2}} ~.
\end{align}
The combined lottery $\natop \TLa + \left(1-\natop\right)\TLc$ has the growth rate
\be
\natop\cdot p_a\frac{\Dx_a}{t_{a,1}}+\natop\cdot \left(1-p_a\right)\frac{\Dx_a}{t_{a,2}} + \left(1-\natop\right)\cdot p_c\frac{\Dx_c}{t_{c,1}}+\left(1-\natop\right)\cdot \left(1-p_c\right)\frac{\Dx_c}{t_{c,2}} =\natop\EAgra + \left(1-\natop\right)\EAgrc ~.
\ee

To show that growth-optimality is continuous, we look for a weight $\natop \in \left[0,1\right]$ for which the growth rate of the combined lottery, $\natop\EAgra + \left(1-\natop\right)\EAgrc$, is exactly the same as \EAgrb. Solving
\be
	\EAgrb = \natop\EAgra + \left(1-\natop\right)\EAgrc
\ee
we get
\be
	\natop = \frac{\EAgrc-\EAgrb}{\EAgrc-\EAgra} ~,
\elabel{p_ens_cont}
\ee
so there is indifference between \TLb and the combination $\natop \TLa + \left(1-\natop\right)\TLc$.

Since $\TLa \prec \TLb$ and $\TLb \prec \TLc$, $\EAgra < \EAgrb$ and $\EAgrb < \EAgrc$. It follows that $\natop\in\left[0,1\right]$.

\paragraph{Continuity in the time approach}

We follow a similar proof for the time approach. In the time approach the time lotteries \TLa, \TLb and \TLc correspond to the following respective growth rates
\begin{align}
	\TAgra		&=	\frac{\Dx_a}{p_a t_{a,1}+ \left(1-p_a\right){t_{a,2}}}	=	\frac{\Dx_a}{\ETa} ~,\\
	\TAgrb		&=	\frac{\Dx_b}{p_b t_{b,1}+ \left(1-p_b\right){t_{b,2}}}	=	\frac{\Dx_b}{\ETb} ~,\\
	\TAgrc	&=	\frac{\Dx_c}{p_c t_{c,1}+ \left(1-p_c\right){t_{c,2}}}	=	\frac{\Dx_c}{\ETc} ~.
\end{align}
The combined lottery $\natop \TLa + \left(1-\natop\right)\TLc$ is the same as in the ensemble approach. However, in the time approach it corresponds to a different growth rate. For calculating this growth rate we assume the lottery is sequentially repeated $N$ times and then evaluate the growth rate after the infinite time limit $N \rightarrow \infty$. We denote by $n_1$ the number of times $\Dx_a$ was paid after time $t_{a,1}$, by $n_2$ the number of times $\Dx_a$ was paid after time $t_{a,2}$, by $n_3$ the number of times $\Dx_c$ was paid after time $t_{c,1}$, by $n_4$ the number of times $\Dx_c$ was paid after time $t_{c,2}$. If follows that the growth rate of the combined lottery is
\begin{align}
	&\lim_{N\to\infty} ~ \frac{\text{Total payment after $N$ rounds}}{\text{Total time elapsed after $N$ rounds}} \\
	&= \lim_{N\to\infty} ~ \frac{\left(n_1+n_2\right)\Dx_a + \left(n_3+n_4\right)\Dx_c}{n_1 t_{a,1} + n_2 t_{a,2} + n_3 t_{c,1} + n_4 t_{c,2}}\\
	&= \lim_{N\to\infty} ~ \frac{\nicefrac{n_1+n_2}{N}\cdot\Dx_a + \nicefrac{n_3+n_4}{N}\cdot\Dx_c}{\nicefrac{n_1}{N}\cdot t_{a,1} + \nicefrac{n_2}{N} \cdot t_{a,2} + \nicefrac{n_3}{N} \cdot t_{c,1} + \nicefrac{n_4}{N}\cdot t_{c,2}}\\
	&= \frac{\natop\Dx_a + \left(1-\natop\right)\Dx_c}{\natop p_a t_{a,1} + \natop \left(1-p_a\right) t_{b,1} + \left(1-\natop\right) p_c t_{c,1} + \left(1-\natop\right) \left(1-p_c\right) t_{c,2}}\\
	&= \frac{\natop\Dx_a + \left(1-\natop\right)\Dx_c}{\natop\ETa + \left(1-\natop\right)\ETc}	~.
\end{align}
To show that growth-optimality is continuous, we look for a weight $\natop \in \left[0,1\right]$ for which the growth rate of the combined lottery, $\frac{\natop\Dx_a + \left(1-\natop\right)\Dx_c}{\natop\ETa + \left(1-\natop\right)\ETc}$, is exactly the same as $\TAgrb = \frac{\Dx_b}{\ETb}$. Solving
\be
\frac{\natop\Dx_a + \left(1-\natop\right)\Dx_c}{\natop\ETa + \left(1-\natop\right)\ETc} = \frac{\Dx_b}{\ETb}
\ee
for $\natop$ yields
\be
\natop = \frac{\Dx_c\ETb - \ETc\Dx_b}{\ETb\left(\Dx_c - \Dx_a\right) + \Dx_b\left(\ETa-\ETc\right)} ~,
\elabel{p_time_cont}
\ee
so there is indifference between \TLb and the combination $\natop \TLa + \left(1-\natop\right)\TLc$.

We now need to show that the value found for $\natop$ is within the interval $\left[0,1\right]$. First, since $\TLb \prec \TLc$, $\TAgrb < \TAgrc$, so $\frac{\Dx_b}{\ETb} < \frac{\Dx_c}{\ETc}$. It follows that the numerator in \eref{p_time_cont} is positive.

Comparing the numerator and the denominator in \eref{p_time_cont} we get:
\begin{align}
	&\Dx_c\ETb - \ETc\Dx_b < \ETb\left(\Dx_c - \Dx_a\right) + \Dx_b\left(\ETa-\ETc\right)\\
	&\iff 0 < \Dx_b\ETa - \Dx_a\ETb\\
	&\iff \frac{\Dx_a}{\ETa} < \frac{\Dx_b}{\ETb}\\
	&\iff \TLa \prec \TLb  ~,
\end{align}
so the denominator is larger than the positive numerator, which renders $\natop\in\left[0,1\right]$.

\qed

\subsection*{Independence}

Given three time lotteries \TLa, \TLb and \TLc, such that $\TLa \prec \TLb$ and a weight $\natop\in\left[0,1\right]$, we would like to show that a growth-optimal decision maker would prefer the combined lottery $\mathcal{B} = \natop \TLb + \left(1-\natop\right)\TLc$ over the combined lottery $\mathcal{A} = \natop \TLa + \left(1-\natop\right)\TLc$.

We define the combined lotteries similarly to the proof of the continuity axiom above. $\mathcal{A}$ is a lottery that pays
\begin{align}
	\Dx_a\text{ at time } t_{a,1}&\text{, with probability } \natop\cdot p_a  	~, \\
	\Dx_a\text{ at time } t_{a,2}&\text{, with probability } \natop\cdot \left(1-p_a\right)  	~, \\
	\Dx_c\text{ at time } t_{c,1}&\text{, with probability } \left(1-\natop\right)\cdot p_c  	~, \\
	\Dx_c\text{ at time } t_{c,2}&\text{, with probability } \left(1-\natop\right)\cdot \left(1-p_c\right)  	~,
\end{align}
and $\mathcal{B}$ is a lottery that pays
\begin{align}
	\Dx_b\text{ at time } t_{b,1}&\text{, with probability } \natop\cdot p_b  	~, \\
	\Dx_b\text{ at time } t_{b,2}&\text{, with probability } \natop\cdot \left(1-p_b\right)  	~, \\
	\Dx_c\text{ at time } t_{c,1}&\text{, with probability } \left(1-\natop\right)\cdot p_c  	~, \\
	\Dx_c\text{ at time } t_{c,2}&\text{, with probability } \left(1-\natop\right)\cdot \left(1-p_c\right)  	~.
\end{align}
%
%
\paragraph{Independence in the ensemble approach}
We begin with the ensemble approach. In the ensemble approach the combined lotteries, $\mathcal{A}$ and $\mathcal{B}$, correspond to the respective growth rates
\begin{align}
	\EAgrA		&=	\natop\EAgra + \left(1-\natop\right)\EAgrc ~,\\
	\EAgrB		&=	\natop\EAgrb + \left(1-\natop\right)\EAgrc ~.
\end{align}
It follows that
\be
\EAgrA < \EAgrB \iff \natop\EAgra + \left(1-\natop\right)\EAgrc < \natop\EAgrb + \left(1-\natop\right)\EAgrc \iff \EAgra < \EAgrb ~,
\ee
and since $\TLa \prec \TLb$, $\mathcal{B}$ is indeed preferred to $\mathcal{A}$.

\qed

\paragraph{Independence in the time approach}
In the time approach the combined lotteries, $\mathcal{A}$ and $\mathcal{B}$, correspond to the following respective growth rates
\begin{align}
	\TAgrA		&=	\frac{\natop\Dx_a + \left(1-\natop\right)\Dx_c}{\natop\ETa + \left(1-\natop\right)\ETc} ~,\\
	\TAgrB		&=	\frac{\natop\Dx_b + \left(1-\natop\right)\Dx_c}{\natop\ETb + \left(1-\natop\right)\ETc} ~.
\end{align}

Assuming $0 < \natop \leq 1$ (as otherwise the combined lottery is not a combined lottery but simply \TLc), it follows that
\begin{align}
	\elabel{independence_cond}
	&\TAgrA < \TAgrB \\
	&\iff \nn \\
	&\natop^2\left[\Dx_a\left(\ETb - \ETc\right) + \Dx_b\left(\ETc - \ETa\right) + \Dx_c\left(\ETa - \ETb\right)\right] + \\
	&\natop\left[\ETc\left(\Dx_a - \Dx_b\right) + \Dx_c\left(\ETb - \ETa\right)\right] < 0\\
	&\iff	\nn\\
	&\natop\left[\Dx_a\left(\ETb - \ETc\right) + \Dx_b\left(\ETc - \ETa\right) + \Dx_c\left(\ETa - \ETb\right)\right] +\\
	&\left[\ETc\left(\Dx_a - \Dx_b\right) + \Dx_c\left(\ETb - \ETa\right)\right] < 0 \\
	&\iff	\nn\\
	&\natop < - \frac{\ETc\left(\Dx_a - \Dx_b\right) + \Dx_c\left(\ETb - \ETa\right)}{\Dx_a\left(\ETb - \ETc\right) + \Dx_b\left(\ETc - \ETa\right) + \Dx_c\left(\ETa - \ETb\right)}\\
	&\iff	\nn\\
	&\frac{\ETc\left(\Dx_b - \Dx_a\right) + \Dx_c\left(\ETa - \ETb\right)}{\ETc\left(\Dx_b - \Dx_a\right) + \Dx_c\left(\ETa - \ETb\right) + \left(\Dx_a\ETb - \Dx_b\ETa\right)} > 1 ~.
\end{align}

To last line is of the form $\frac{a}{a+b}$, where $a = \ETc\left(\Dx_b - \Dx_a\right) + \Dx_c\left(\ETa - \ETb\right)$, and $b = \Dx_a\ETb - \Dx_b\ETa$. To show that the independence condition $\TAgrA < \TAgrB$ in \eref{independence_cond} is fulfilled, it is therefore required that $a > 0$, $b < 0$ and $a+b > 0$. While it is guaranteed that $b < 0$ from $\TLa \prec \TLb$, it is possible that $a < 0$. In such a case, there might be weights $\natop$ and time lotteries \TLc such that $\TLa \prec \TLb$ but $\natop \TLb + \left(1-\natop\right)\TLc \prec \natop \TLa + \left(1-\natop\right)\TLc$. In such cases, the independence axiom is violated.

For example, if we consider the time lotteries detailed in \tref{counterexample} we get that $\TLa \prec \TLb$, since $\TAgra=6.67~\nicefrac{\$}{\text{sec}}$ and $\TAgrb=8.42~\nicefrac{\$}{\text{sec}}$. However, for $\natop=0.1$, we get that the combined lottery $\mathcal{A} = \natop \TLa + \left(1-\natop\right)\TLc$ corresponds to a growth rate of $\TAgrA = 0.87~\nicefrac{\$}{\text{sec}}$, and $\mathcal{B} = \natop \TLb + \left(1-\natop\right)\TLc$ to $\TAgrB = 0.82~\nicefrac{\$}{\text{sec}}$.

\begin{table}[!htb]
\ra{1.25}
\scriptsize
\centering
\captionof{table}{A counterexample for the independence criterion in the time approach.}\tlabel{counterexample}
\begin{tabular}{ccccc}
\addlinespace
\toprule[1.5pt]
\addlinespace
Time lottery & \makecell{$t_1$\\ earlier payment time\\ (sec)} & \makecell{$t_2$\\ later payment time\\ (sec)} & \makecell{$p$\\ earlier payment probability} & \makecell{$\Dx$\\ payment\\ ($\$$)} \\
\midrule[1.5pt]
\TLa		& 1		&	2	&	0.5	&	10 \\
\TLb		& 0.5	&	2	&	0.7	&	8 \\
\TLc		& 2		&	4	&	0.3	&	2 \\
\bottomrule[1.5pt]
\end{tabular}
\end{table}

Yet, as commonly described in the literature, the standard case is a comparison between time lotteries with the same payment. In the context of the time lotteries described herein, this means $\Dx = \Dx_a = \Dx_b = \Dx_c$. It follows that
\be
\elabel{independence_cond2}
\TAgrA < \TAgrB \iff \natop \Dx\left(\ETb - \ETa\right) < 0 ~.
\ee

Under the assumption of equal payments $\TLa \prec \TLb$ implies $\nicefrac{\Dx}{\ETa} < \nicefrac{\Dx}{\ETb}$ and $\ETa > \ETb$. Thus, indeed, $\TAgrA < \TAgrB$ and independence is satisfied.

\qed

\end{document}